\documentclass[cdm]{ipbook}
\usepackage{amssymb, latexsym, amsmath}
\usepackage{graphicx}

\usepackage{dcolumn}% Align table columns on decimal point
\usepackage{bm, amsmath,amssymb, mathrsfs, comment, color, mathtools}% bold math

\newtheorem{theorem}{Theorem}
\newtheorem{definition}[theorem]{Definition}

 \numberwithin{equation}{section}

\newcommand{\R}{\mathbb R}

\newcommand{\na}{\nabla}
\newcommand{\lt}{\left}
\newcommand{\rt}{\right}

\newcommand{\rw}{\rightarrow}

\startlocaldefs
\endlocaldefs

\firstpage{1}
\lastpage{1}

\title[[Angular momentum and supertranslation]{Angular momentum and supertranslation in general relativity}
\author[]{{Mu-Tao Wang}}
\address{Department of Mathematics\\
Columbia University\\ New York\\ NY 10027\\ USA}
\email{mtwang@math.columbia.edu}\thanks{This material is based upon work supported by the
National Science Foundation under Grant No. DMS-1810856 and DMS-2104212}

\begin{document}

\begin{abstract}

How does one measure the angular momentum carried away by gravitational radiation during the merger of a binary black hole? This has been a subtle issue since the 1960’s due to the discovery of ``supertranslation ambiguity”: the angular momentums recorded by two distant observers of the same system may not be the same.
In this talk, I shall describe how the theory of quasilocal mass and optimal isometric embedding identifies a new definition of angular momentum that is free of any supertranslation ambiguity. This is based on joint work with Po-Ning Chen, Jordan Keller, Ye-Kai Wang, and Shing-Tung Yau.\end{abstract}

\maketitle

\tableofcontents

\section{\label{sec:level1}Introduction}

In the first observation of gravitational waves by LIGO and Virgo \cite{Abbott2016}, the event GW150914 corresponds to a binary black hole merger.  

\begin{figure}[h]
\caption{Binary black hole coalescence  \cite[figure 12.1 on page 396]{BS}}
\centering
\includegraphics[scale=0.2]{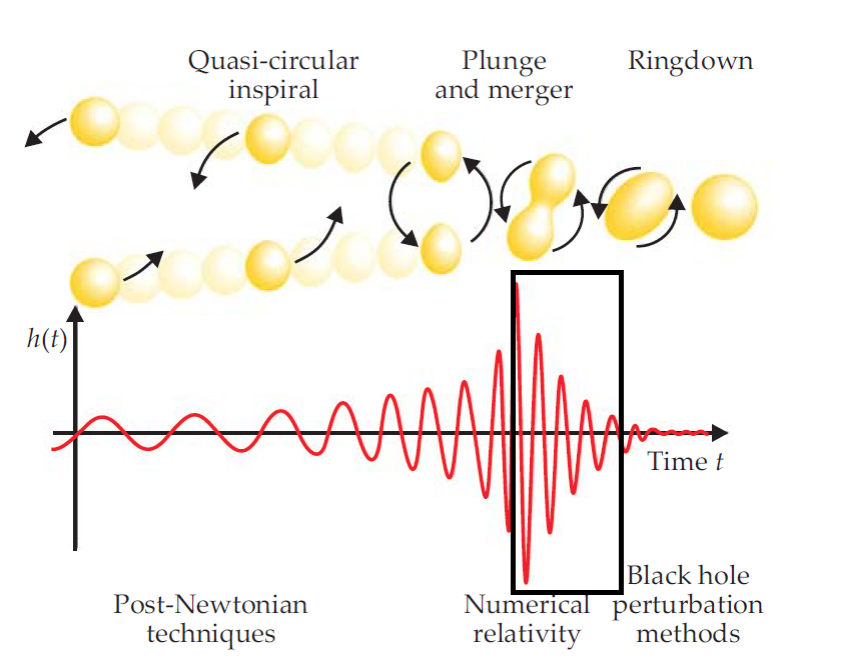}\end{figure}

The initial black hole masses are 36 units and 29 units and the final black hole mass is 62 units, with 3 units of mass radiated away in gravitational waves. The initial two black holes rotate about each other and then settle down to a single rotating black hole. One naturally wonders how much ``angular momentum" radiated away  in  gravitational waves.  This turned out to be a more subtle and challenging question due to the presence of ``supertranslation ambiguity". The amount of mass radiated away is ``supertranslation invariant", i.e. any two observers record the same amount of mass flux. On the other hand, two observers of the same system may record different amount of ``angular momentum" flux. A goal of these talks is  to explain what supertranslation is and where this ambiguity comes from.

There have been a lot of efforts to resolve the supertranslation ambiguity since the 1960's. One approach attempts to single out a ``best observer". Another approach is to find a new definition of angular momentum. Last year, the first such ambiguity free definition of angular momentum \cite{CWWY1, CKWWY} was discovered and it is the purpose of these talks to present this new definition.

\section{Expository lecture for general audience}
\subsection{Future null infinity and supertranslation}
A key axiom of general relativity is that the speed of light is a universal constant. Fixing this constant to be one in a spacetime depiction where the vertical direction pointing upward is the future time direction, light travels in a 45 degree upward direction along a so-called ``null geodesic".  Light rays/null geodesics emanating from an astronomical event (such as a binary black hole merger) in all directions at an instant form a ``null hypersurface". 
These null hypersurfaces can be labelled by the ``retarded time" that is usually denoted by $u$.

\begin{figure}[h]
\caption{A null hypersurface that corresponds to $u=u_1$}
\centering
\includegraphics[scale=0.15]{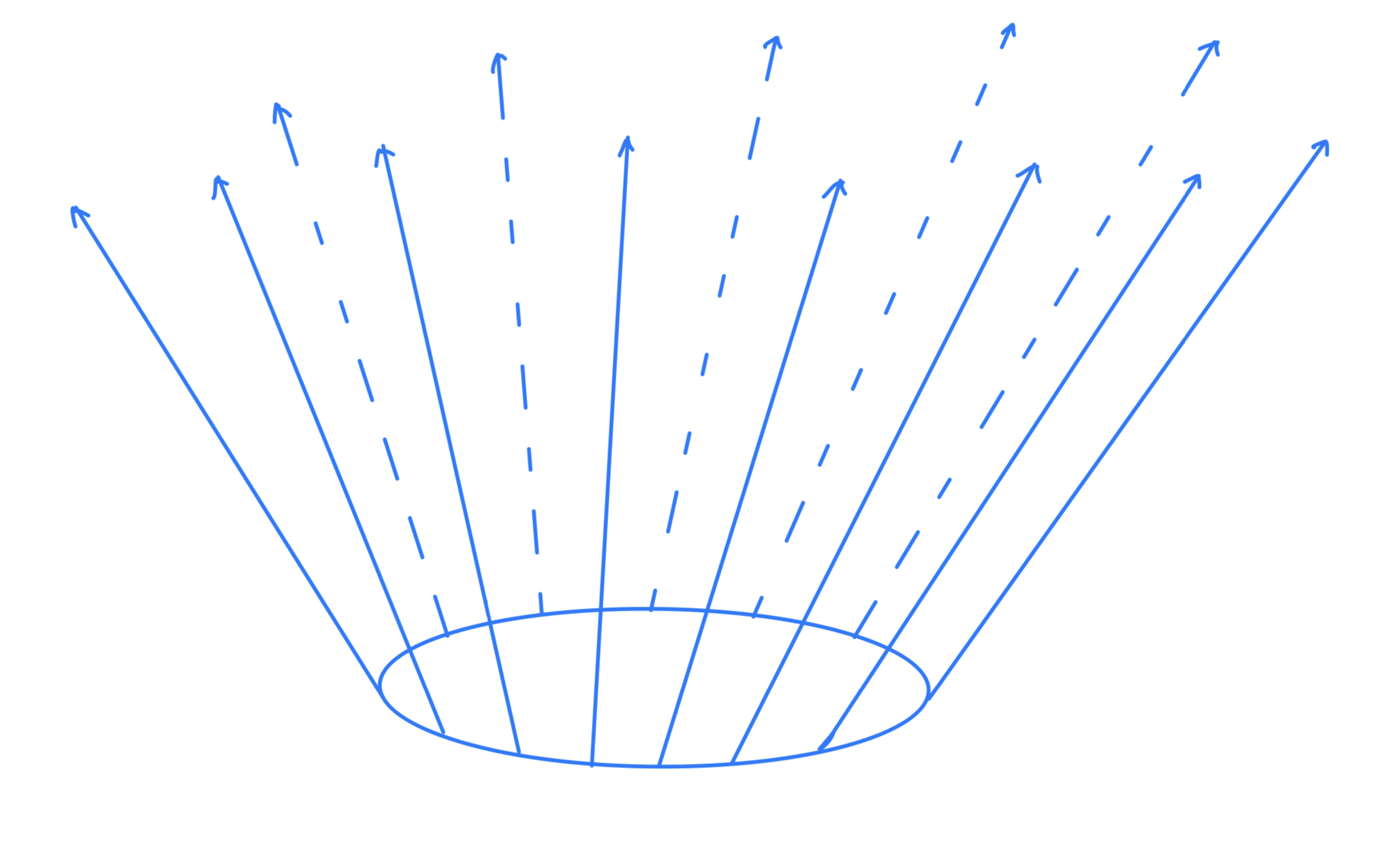}\end{figure}

\newpage
\begin{figure}[h]
\caption{Two null hypersurfaces that correspond to $u=u_0$ and $u=u_1$}
\centering
\includegraphics[scale=0.14]{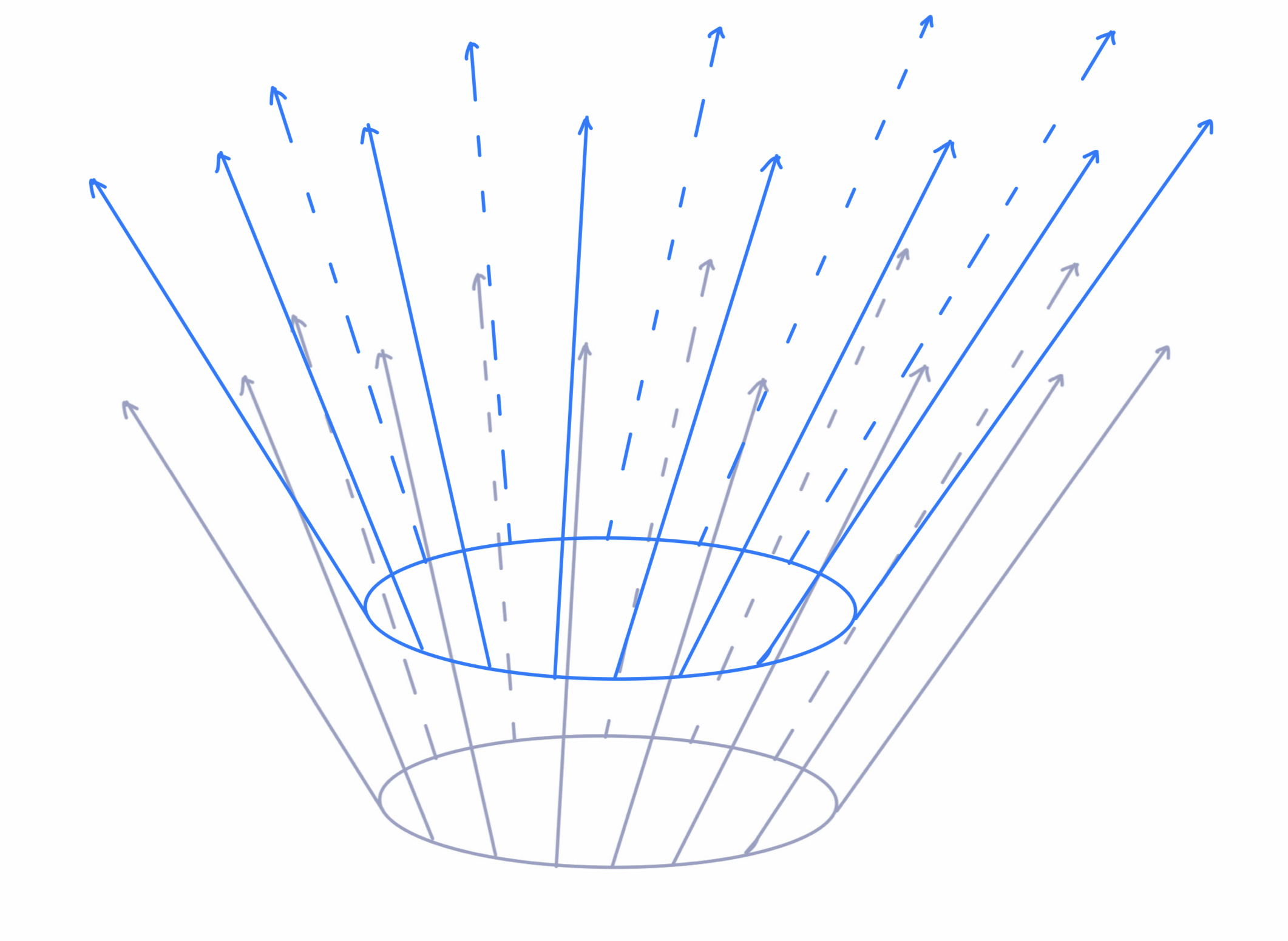}\end{figure}

\newpage

An idealized distant observer is situated at the collection of the ends of these null geodesics/hypersurfaces, called future null infinity and denoted as $\mathscr{I}^+$.
For an isolated gravitating system, the gravitational field is weak at future null infinity and the system corresponds to an asymptotically flat spacetime, i.e. a spacetime that approaches the flat Minkowski spacetime at a faraway distance.  Penrose's conformal compactification \cite{Penrose3, Penrose4} identifies the future null infinity of an asymptotically flat spacetime as the product of a two-sphere $S^2$ (all directions) and a real line $(-\infty, \infty)$ parametrized by the retarded time $u$. $u=-\infty$ corresponds to spacelike infinity and $u=\infty$ corresponds to timelike infinity.

\begin{figure}[h]
\caption{Future null infinity $\mathscr{I}^+$}
\centering
\includegraphics[scale=0.5, angle=270]{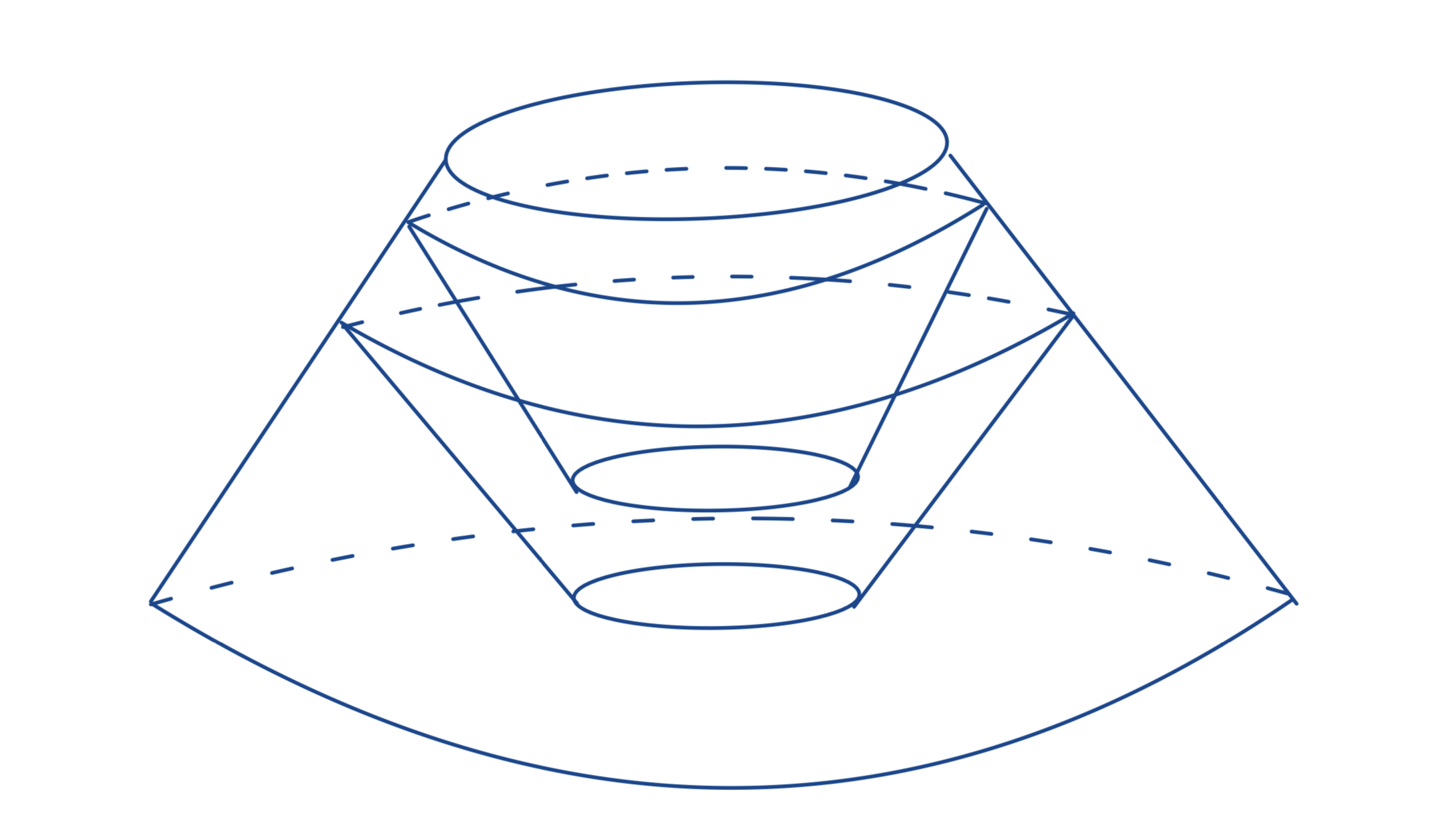}\end{figure}

\newpage

In the Minkowski spacetime with coordinates $(t, x, y, z)$ and the spacetime metric \[-dt^2+dx^2+dy^2+dz^2,\] $u$ can be taken to be $u=t-r$ for $r=\sqrt{x^2+y^2+z^2}$. The level sets of $u$ are the null hypersurfaces described as above and cut the future null infinity at two-sphere's. 

When we translate the coordinate system $(t, x, y, z)$, we obtain a new retarded time $\bar{u}=\bar{t}-\bar{r}$ that is of the form
\[ \bar{u}=u+\alpha_0+\alpha_1 \sin\theta\sin\phi+\alpha_2 \sin\theta \cos\phi+\alpha_3 \cos\theta+ O(\frac{1}{r}),\] where $x=r\sin\theta\sin\phi, y=r\sin\theta\cos\phi, z=\cos\theta$ and $\alpha_0, \alpha_1, \alpha_2, \alpha_3$ are the translation constants. See the appendix for the calculation.

It was discovered by Bondi-Metzner-Sachs (BMS) \cite{BVM, Sachs}  in the 1960's that one can take $\bar{u}$ of the form 
\[ \bar{u}=u-f(\theta, \phi)+O(\frac{1}{r})\] for any smooth function $f$ on $S^2$, such that the level sets of $\bar{u}$ are still null hypersurfaces and the spacetime metric only differs from the Minkowski metric by a term of the order $O(\frac{1}{r})$. This is exactly the form of an asymptotically flat spacetime considered by Bondi et al.

For a general asymptotically flat spacetime, one can perform such a coordinate change as well. The new retarded time $\bar{u}$ is indistinguishable from the original retarded time $u$ for a general isolated gravitating system with radiation. The symmetry of future null infinity is thus the infinite dimensional BMS group, which is much larger than the 10 dimensional Poincar\'e group (the symmetry group of the Minkowski spacetime).

 When  $f$ is of harmonic mode $\ell=1$ (i.e. of the form $\alpha_0+\alpha_1 \sin\theta\sin\phi+\alpha_2 \sin\theta \cos\phi+\alpha_3 \cos\theta$), the coordinate change of $u$ to ${\bar u}=u-f$ is an ``ordinary" translation.
The change by a higher mode function on the two-sphere is thus called a ``supertranslation". Ordinary translations are part of the Poincar\'e group, while supertranslations are not. 
In general, the supertranslation ambiguity corresponds to different yet indistinguishable choices of foliations of future null infinity, or different idealized distant observers of the isolated gravitating system. 

\begin{figure}[h]
\caption{Consider the future null infinity as a lampshade. Different retarded times correspond to different markings on the lampshade}
\centering
\includegraphics[scale=0.2]{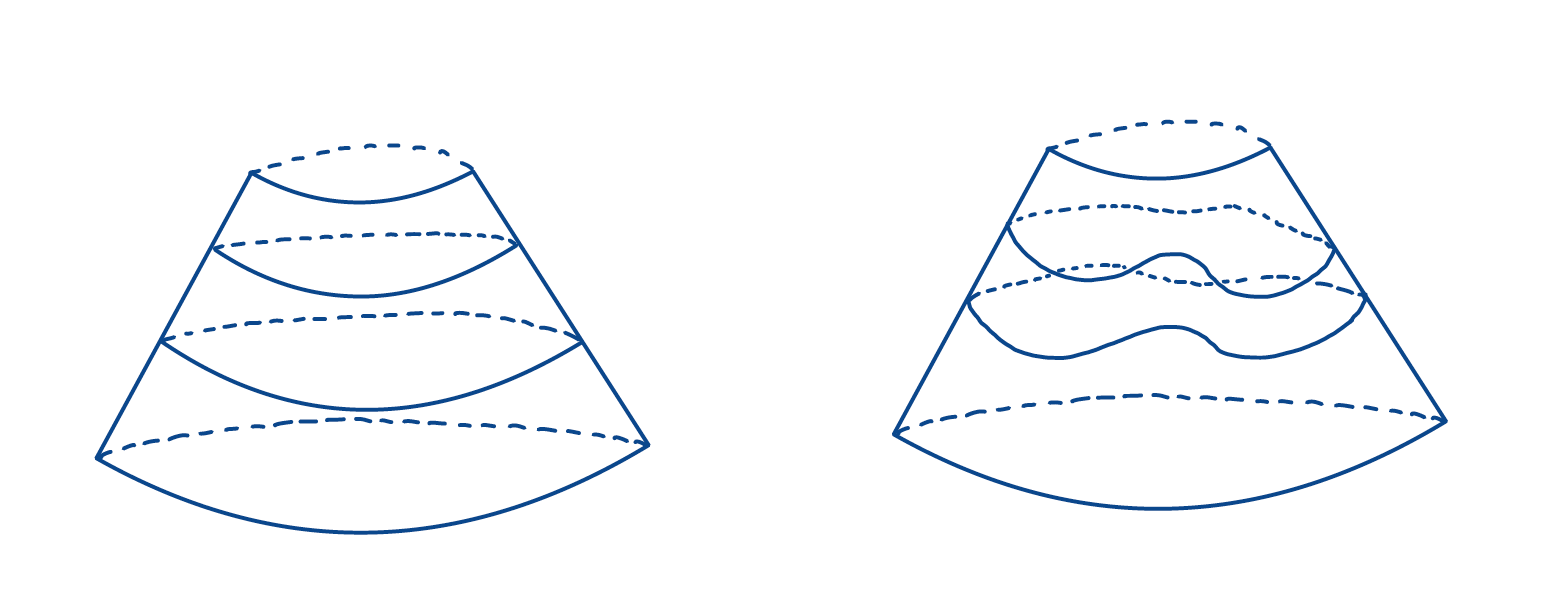}\end{figure}

\newpage

\subsection{Bondi energy/mass at future null infinity}

For a fixed retarded time $u$, to each $u$ slice (or level set of $u$) $u=u_0$, there is an associated notion of Bondi mass/energy $E(u_0)$. The Bondi mass satisfies the well-known mass loss formula, i.e. $E(u_0)\geq E(u_1)$ if $u_0\leq  u_1$, which manifests how mass radiated away along the null directions. This is one of the first theoretical verifications of the wave nature of gravitation.

\begin{figure}[h]
\caption{Mass-loss due to radiation from Penrose \cite{Penrose2}}
\centering
\includegraphics[scale=0.14]{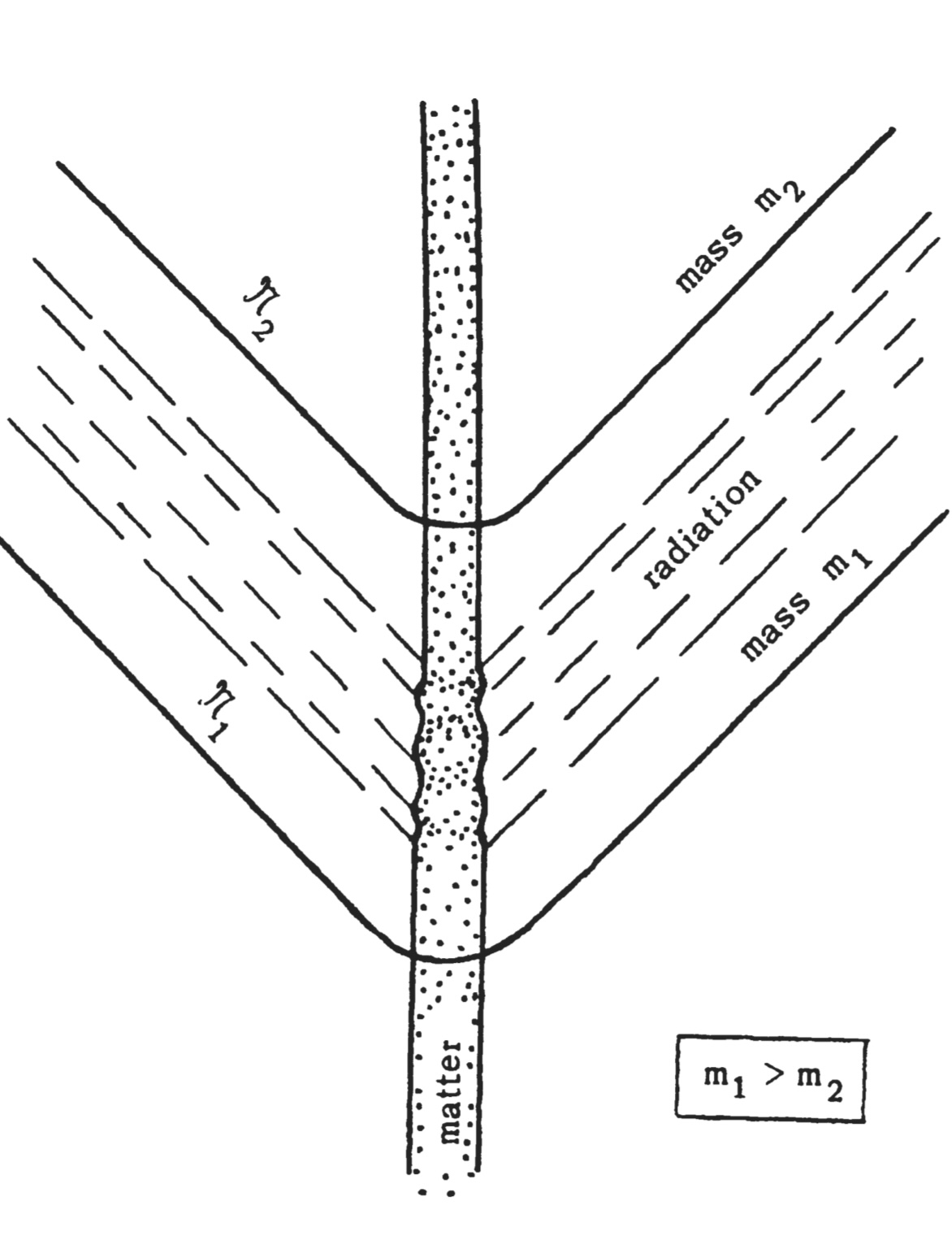}\end{figure}

 It is also proved to be positive and thus an isolated system cannot radiate away more mass than it has initially (the Arnowitt-Deser Misner (ADM) mass at $u=-\infty$).

 Moreover, the mass radiated away or the total flux of mass is supertranslation invariant, i.e. 

\begin{equation}\label{super-inv} \lim_{u\rightarrow \infty} E(u)-\lim_{u\rightarrow -\infty} E(u)= \lim_{\bar{u} \rightarrow \infty} E(\bar{u})-\lim_{\bar{u}\rightarrow -\infty} E(\bar{u}),\end{equation} if $u$ and $\bar{u}$ are related by a supertranslation or an ordinary translation. 

\subsection{The story of angular momentum}

On the other hand, the story for angular momentum is quite different. 

 In classical mechanics, the angular momentum of a particle is defined as
\[ J = m \mathbf{r} \times \mathbf{r}'. \]  If we translate the origin (or the coordinate system) to $\tilde{O}$, the angular momentum $J$ gets shifted by the linear momentum $\mathbf{p}$
\[ J_{\tilde{O}} = m (\mathbf{r} - \mathbf{a}) \times \mathbf{r}' = J - \mathbf{a} \times \mathbf{p}.\]

\begin{figure}[h]
\caption{Angular momentum gets shifted by the linear momentum}
\centering
\includegraphics[scale=0.09]{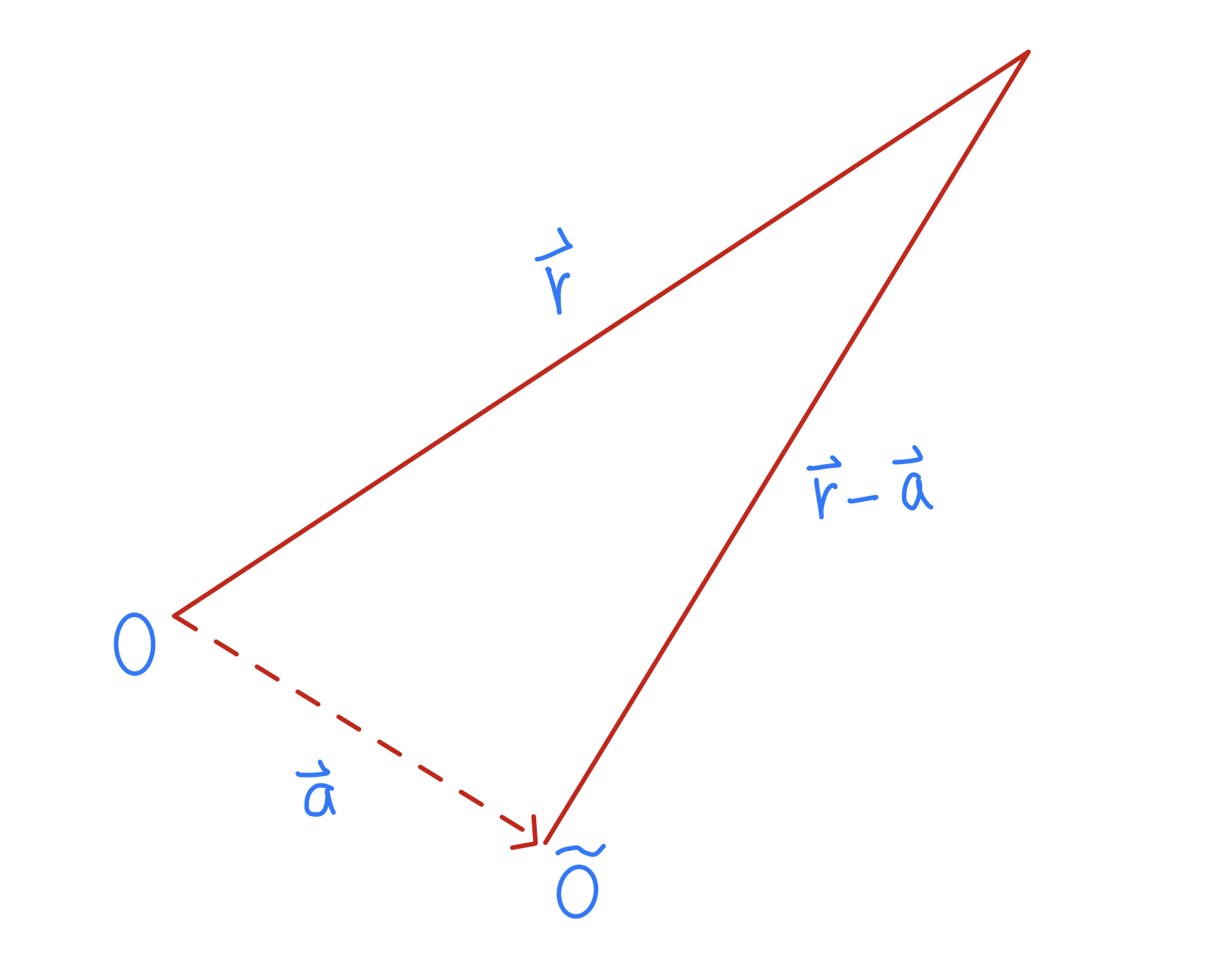}\end{figure}

In special relativity in which one deals with matter fields on the flat Minkowski spacetime. Each matter field has an energy-momentum density tensor. Pairing the energy-momentum density tensor with one of the ten dimensional Killing fields gives a conserved quantity. In particular, a rotation Killing field such as $x\partial_y-y\partial_x$ gives an angular momentum.

 In general relativity, the spacetime is a Lorenzian 4-manifold that satisfies the Einstein equation. There is no energy-momentum density for gravitation due to Einstein's equivalence principle (or the energy of the gravitational field cannot be localized, see Chapter 20 of Misner-Thorne-Wheeler, \cite{MTW}. In addition, there is no Killing field on a general spacetime. However, near the infinity of an isolated gravitating system/asymptotically flat spacetime, the spacetime is asymptotic to the Minkowkian spacetime. Defining angular momentum at null infinity involves finding an ``asymptotic" energy-momentum density and an ``asymptotic" Killing field. However, this becomes very subtle when radiation is present.

 There have been various proposals of angular momentum at future null infinity since 1960's. Different approaches (Hamiltonian, spinor-twistor, Komar type etc.) have led to different definitions such as Newmann-Penrose \cite{NP2}, Winicour-Tamburino \cite{WT}, Bramson 1975 \cite{Bramson}, Ashtekar-Hansen 1978 \cite{AH}, Penrose 1982 \cite{Penrose1},  Ludvigsen-Vickers 1983 \cite{LV}, 
Dray-Streubel 1984 \cite{DS}, Moreschi 1986 \cite{Moreschi}, Dougan-Mason 1991 \cite{DM}, Rizzi 1997 \cite{Rizzi}, Chru\'sciel-Jezierski-Kijowski 2002 \cite{CJK}, Helfer 2007 \cite{Helfer1, Helfer2}, Barnich-Troessaert 2011\cite{BT}, Hawking-Perry-Strominger 2017 \cite{HPS}, Javadinezhad-Kol-Porrati 2019 \cite{JKP19}, Klainerman-Szeftel 2019 \cite{KS}, Comp\`ere-Oliveri-Seraj 2020 \cite{COS} etc. 

Our approach to resolve the ambiguity is to find a new definition of angular momentum $J(u)$ whose total flux satisfies the supertranslation invariance/covariance property
To obtain an angular momentum for a null hypersurface $u=u_0$, we consider a family of two-surfaces parametrized by $r$ ($4\pi r^2$ is the area of such a surface) such that $r\rightarrow \infty$ near future null infinity.

\begin{figure}[h]
\caption{A null hypersurface $u=u_0$ foliated by $r$ level sets}
\centering
\includegraphics[scale=0.6]{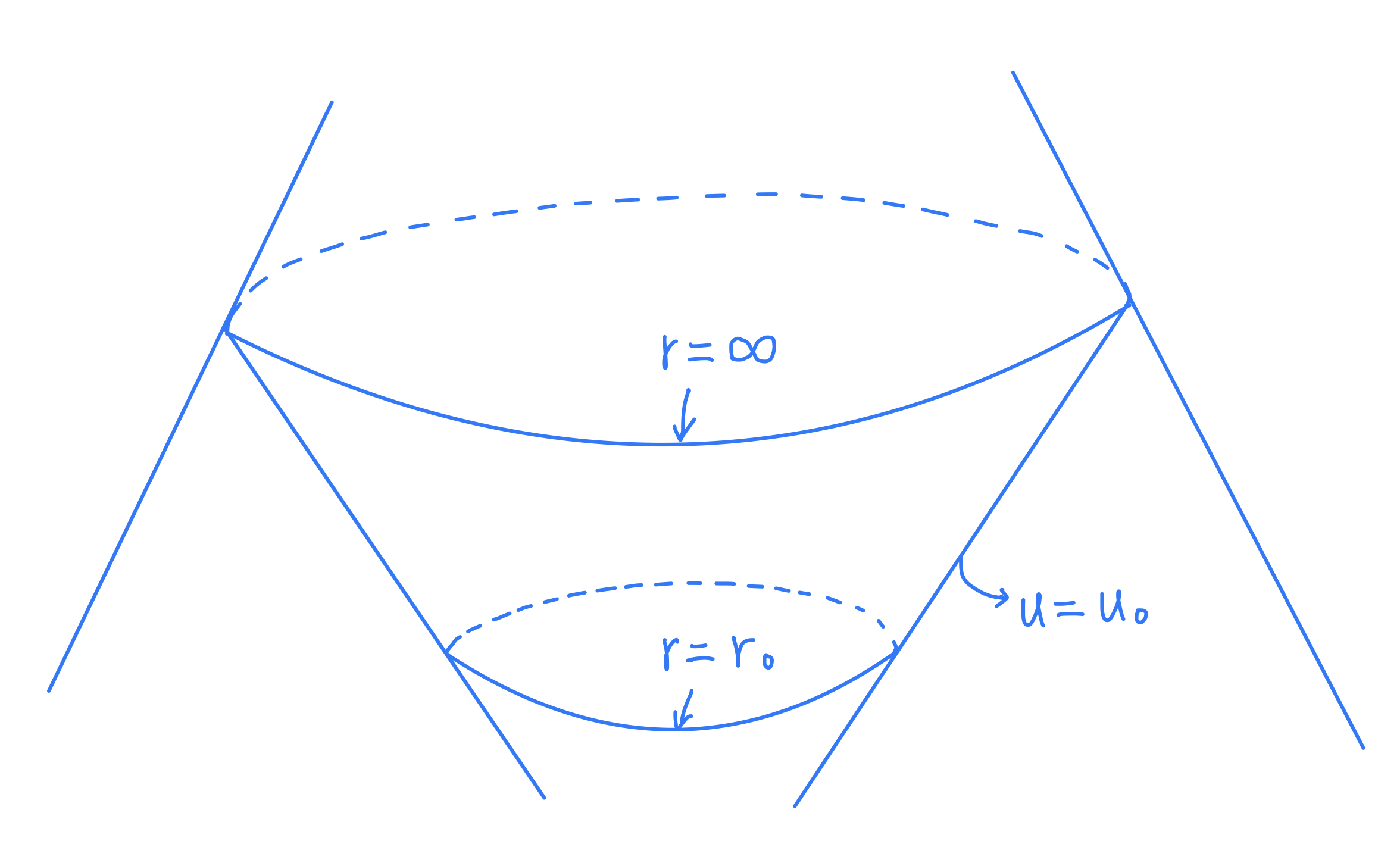}\end{figure}

Therefore, the null hypersurface $u=u_0$ is foliated by 2-surfaces given by $u=u_0, r=r_0$, on which a notion of quasilocal angular momentum $J(u_0, r_0)$ was defined by Chen-Wang-Yau in 2015 \cite{CWY3, CWY4}. These surfaces approach the cut at null infinity as $r\rightarrow \infty$ and we define $J(u)=\lim_{r\rightarrow\infty} J(u, r)$.  It is proved in \cite{CWWY1, CKWWY} that the total flux of $J(u)$ satisfies the desired invariant property, i.e. 

\[ \lim_{u\rightarrow \infty} J(u)-\lim_{u\rightarrow -\infty} J(u)= \lim_{\bar{u} \rightarrow \infty} J(\bar{u})-\lim_{\bar{u}\rightarrow -\infty} J(\bar{u}),\] if $u$ and $\bar{u}$ are related by a supertranslation.

On the other hand, when $u$ and $\bar{u}$ are related by an ordinary translation, the corresponding total fluxes are related by an exact analogue 
of the transformation formula in classical mechanics. 
The supertranslation invariance/covariance property will be explained in terms of the Bondi-Sachs coordinate system, in which the supertranslation ambiguity was first discovered. Such ambiguity is indeed ubiquitous in any description of future null infinity. See \cite{CWWY3} for supertranslation ambiguity in double null gauge near null infinity.

\section{In-depth discussion}
\subsection{Quasilocal mass and angular momentum }

The definitions of conserved quantities such as mass and angular momentum have been among the most difficult problems since the genesis of general relativity. According to Einstein's equivalence principle, there is no density for gravitation and no canonical coordinate system for spacetime. The issue is further complicated by the nonlinear nature of Einstein's eponymous equation.

 The notion of quasilocal mass is attached to a 2-dimensional closed surface $\Sigma$ which bounds a spacelike region in spacetime. $\Sigma$ is assumed to be a topological 2-sphere, but with different intrinsic geometry and extrinsic geometry, we expect to read off the effect of gravitation in the spacetime vicinity of the surface. Suppose the surface is spacelike, i.e. the induced metric $\sigma$ is Riemannian.  An essential part of the extrinsic geometry is measured by the mean curvature vector field $\bf{H}$ of $\Sigma$. $\bf{H}$ is a normal vector field of the surface such that the null expansion along any null normal direction $\ell$ is given by the pairing $\langle {\bf H}, \ell\rangle$ of
$\bf{H}$ and $\ell$.

In \cite{Wang-Yau1}, Wang-Yau proposed the following definition of quasilocal mass which depends only on $\sigma$ and $\bf{H}$ of a 2-surface $\Sigma$ in spacetime.  To evaluate the quasilocal mass of $\Sigma$ with the physical data $(\sigma, \bf{H})$, one first solves the optimal isometric embedding equation, see \eqref{oiee} below, which gives an embedding of $\Sigma$ into the Minkowski spacetime with the image
surface $\Sigma_0$ that has the same 
induced metric as $\Sigma$, i.e. $\sigma$. One then compares the extrinsic geometries of $\Sigma$ and $\Sigma_0$ and evaluates the quasilocal mass from $\sigma, \bf{H}$ and $\bf{H_0}$.

\begin{figure}[h]
\caption{Isometric embedding into the Minkowski spacetime}
\centering
\includegraphics[scale=0.2]{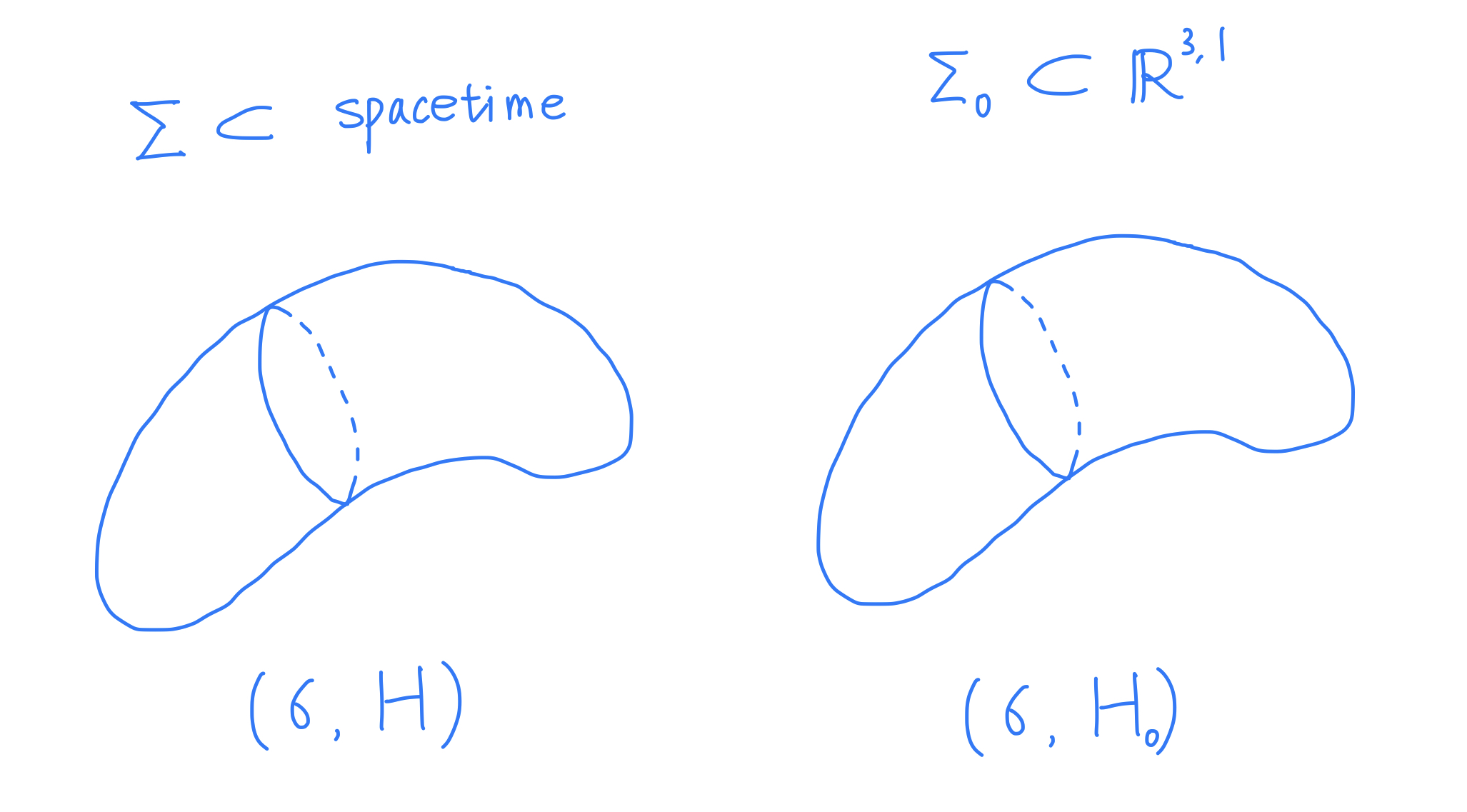}\end{figure}

Assuming the mean curvature vector ${\bf H}$ is spacelike, the physical surface $\Sigma$ with physical data $(\sigma, \bf{H})$ gives $(\sigma, |\bf{H}|, \alpha_{\bf {H}})$ where $|{\bf H}|>0$ is the Lorentz norm of $\bf{H}$ and $\alpha_{\bf H}$ is the connection one-form determined by $\bf{H}$. Given an isometric embedding $X:\Sigma\rightarrow \R^{3,1}$ of $\sigma$. Let $\Sigma_0$ be the image $X(\Sigma)$ and $(\sigma, |\bf{H}_0|, \alpha_{\bf {H}_0})$ be the corresponding data of $\Sigma_0$ (${\bf H_0}$ is again assumed to be spacelike).

Let $T$ be a future timelike unit Killing field of $\R^{3,1}$ and define $\tau=-\langle X, T\rangle$ as a function on $\Sigma$. Define a function $\rho$ and a 1-form $j_a$ on $\Sigma$:
  \[ \begin{split}\rho &= \frac{\sqrt{|{\bf H}_0|^2 +\frac{(\Delta \tau)^2}{1+ |\nabla \tau|^2}} - \sqrt{|{\bf H}|^2 +\frac{(\Delta \tau)^2}{1+ |\nabla \tau|^2}} }{ \sqrt{1+ |\nabla \tau|^2}}\\
 j_a&=\rho {\nabla_a \tau }- \nabla_a \left( \sinh^{-1} (\frac{\rho\Delta \tau }{|{\bf H}_0||{\bf H}|})\right)-(\alpha_{{\bf H}_0})_a + (\alpha_{{\bf H}})_a, \end{split}\] where $\nabla_a$ is the covariant derivative with respect to the metric $\sigma$, $|\nabla \tau|^2=\nabla^a \tau\nabla_a \tau$ and $\Delta \tau=\nabla^a\nabla_a \tau$. 
$\rho$ is the quasilocal mass density and $j_a$ is the quasilocal momentum density. A full set of quasilocal conserved quantities was defined in \cite{CWY3, CWY4} using $\rho$ and $j_a$. 
 
The optimal isometric embedding equation for $(X, T)$ is 
\begin{equation}\label{oiee} \begin{cases}
\langle dX, dX\rangle&=\sigma\\
\nabla^a j_a&=0.
\end{cases}\end{equation}
The first equation is the isometric embedding equation into the Minkowski spacetime and the second one is the Euler-Lagrange equation of the quasilocal energy $E(\Sigma, \tau)$ \cite{Wang-Yau1,Wang-Yau2} in the space of isometric embeddings.
The quasi-local mass for the optimal isometric embedding $(X, T)$ is defined to be \[E(\Sigma, X, T)=\frac{1}{8\pi}\int_\Sigma \rho.\]
 It is shown in \cite{Wang-Yau1, Wang-Yau2} that $E(\Sigma, X, T)$ is { positive in general}, and { zero for surfaces in the Minkowski spacetime}. 

The theory of quasilocal mass and optimal isometric embedding was employed by Chen-Wang-Yau in \cite{CWY3, CWY4} to define quasilocal conserved quantities. For an optimal isometric embedding $(X, T)$, by restricting a rotation (or boost) Killing field $K$ of $\R^{3,1}$ to $\Sigma_0=X(\Sigma)\subset \R^{3,1}$, the quasi-local conserved quantity is defined to be: 
\[-\frac{1}{8\pi} \int_\Sigma \langle K, T\rangle \rho+(K^\top)^a  j_a ,\] where $K^\top$ is the component of $K$ that is tangential to $\Sigma_0$.
In particular, $K=x^i\partial_j-x^j \partial_i, i<j$ defines an angular momentum with respect to $\partial_t$. Here $(t, x^i)$ and $(\partial_t, \partial_i)$ are standard coordinates and coordinate vectors of the Minkowski spacetime.

The image of the optimal isometric embedding $\Sigma_0$ is essentially the ``unique" surface in the Minkowski spacetime that best matches the physical surface $\Sigma$. 
If the original surface $\Sigma$ happens to be a surface in the Minkowski spacetime, the above procedure identifies $\Sigma_0=\Sigma$ up to a global isometry. The theorems in \cite{CWY2} allow us to solve the optimal isometric embedding system for configurations that limit
to a surface in the Minkowski spacetime. This is in particular sufficient for calculations at infinity of an isolated system.

For surfaces near the null infinity, the system of equations  \eqref {oiee} has a unique solution, which can be calculated explicitly \cite{CWY1} and numerically. The limit of the Wang-Yau quasilocal energy momentum in Bondi-Sachs coordinates was evaluated in \cite{CWY1}.

 Transplanting a rotation Killing field $K$ of $\R^{3,1}$ through this unique solution, we define the quasilocal angular momentum on $\Sigma$ as
\[
\int_\Sigma j(K^T).
\]

The limit of the Chen-Wang-Yau quasilocal  angular momentum defines the total angular momentum at null infinity. This can be carried out for very general spacetimes without any peeling structure at null infinity.

\subsection{Bondi-Sachs coordinates and supertranslation invariance of the Bondi-Sachs energy-momentum}

An important description of future null infinity is the Bondi-Sachs coordinate system \cite{BVM, Sachs} in terms of which the action of the BMS group was originally formulated.

\begin{figure}[h]
\caption{Future null infinity $\mathscr{I}^+$ from $u=-\infty$ ($i^0$) to $u=\infty$ ($i^+$)}
\centering
\includegraphics[scale=0.4]{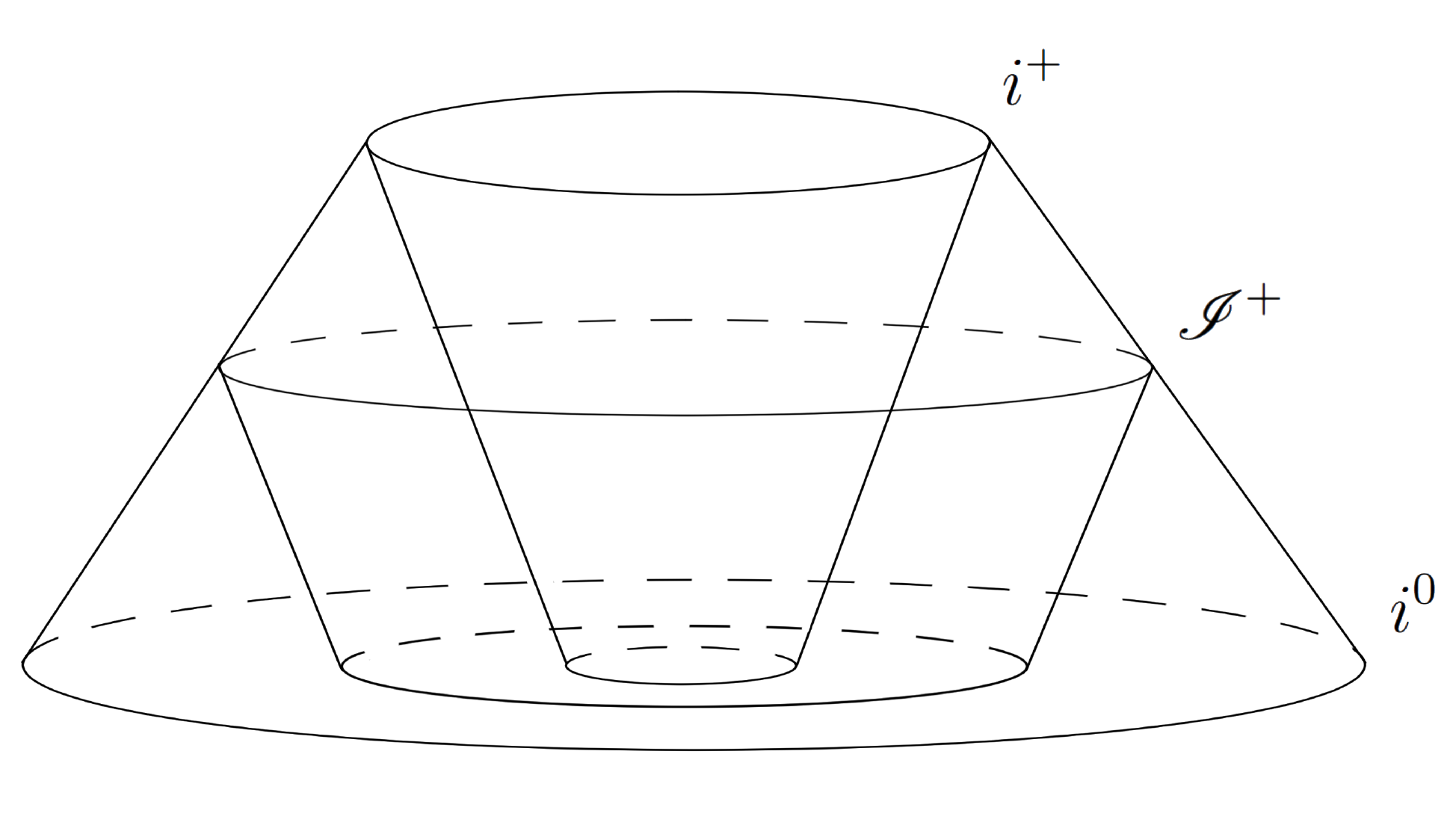}\end{figure}

In a  Bondi-Sachs coordinate system $(u, r,  x^2, x^3)$, the physical spacetime metric takes the form
\begin{equation}\label{spacetime_metric} -UV du^2-2U dudr+r^2 h_{AB}(dx^A+W^A du)(dx^B+W^B du), A, B=2, 3.\end{equation} The future null infinity $\mathscr{I}^+$ corresponds to the idealized null hypersurface $r=\infty$ and can be viewed $\mathscr{I}^+=I\times S^2$ with coordinates $(u, x)$ where $u\in I$ and $x=(x^2, x^3)\in S^2$, the unit sphere. The outgoing radiation condition \cite{Sachs} and the vacuum Einstein equation imply the following expansions in inverse powers of $r$:
 \[\begin{split}
V&=1-\frac{2m}{r}+  O(r^{-2}),\\
W^A&= \frac{1}{2r^2} \na_B C^{AB} + O(r^{-3}),\\
h_{AB}&={\sigma}_{AB}+\frac{C_{AB}}{r}+ O(r^{-2}),\end{split} \] where $\sigma_{AB}(x)$ is a standard round metric on $S^2$ and $\nabla_A$ denotes the covariant derivative with respect to $\sigma_{AB}$. The indices are contracted, raised, and lowered with respect to the metric $\sigma_{AB}$.
Defining on $\mathscr{I}^+$ $(r=\infty)$ are the {\it mass aspect} $m=m(u, x)$, the {\it angular aspect} $N_A = N_A(u, x)$, and the {\it shear} $C_{AB}=C_{AB}(u, x)$ of this Bondi-Sachs coordinate system.  We also define the {\it news} $N_{AB} = \partial_u C_{AB}$.  

A {\it supertranslation} is a change of coordinates $(\bar{u}, \bar{x})\rightarrow (u, x)$ at $\mathscr{I}^+$ such that 
\begin{equation}\label{coord_change} u = \bar u + f (x), x=\bar{x}\end{equation} for a smooth function $f$ that is defined on $S^2$. It can be shown that this can be extended to $(\bar{u},\bar{r}, \bar{x})\rightarrow (u, r, x)$ of two Bondi-Sachs coordinate systems.

\begin{figure}[h]
\caption{Two retarded times related by a supertranslation}
\centering
\includegraphics[scale=0.13]{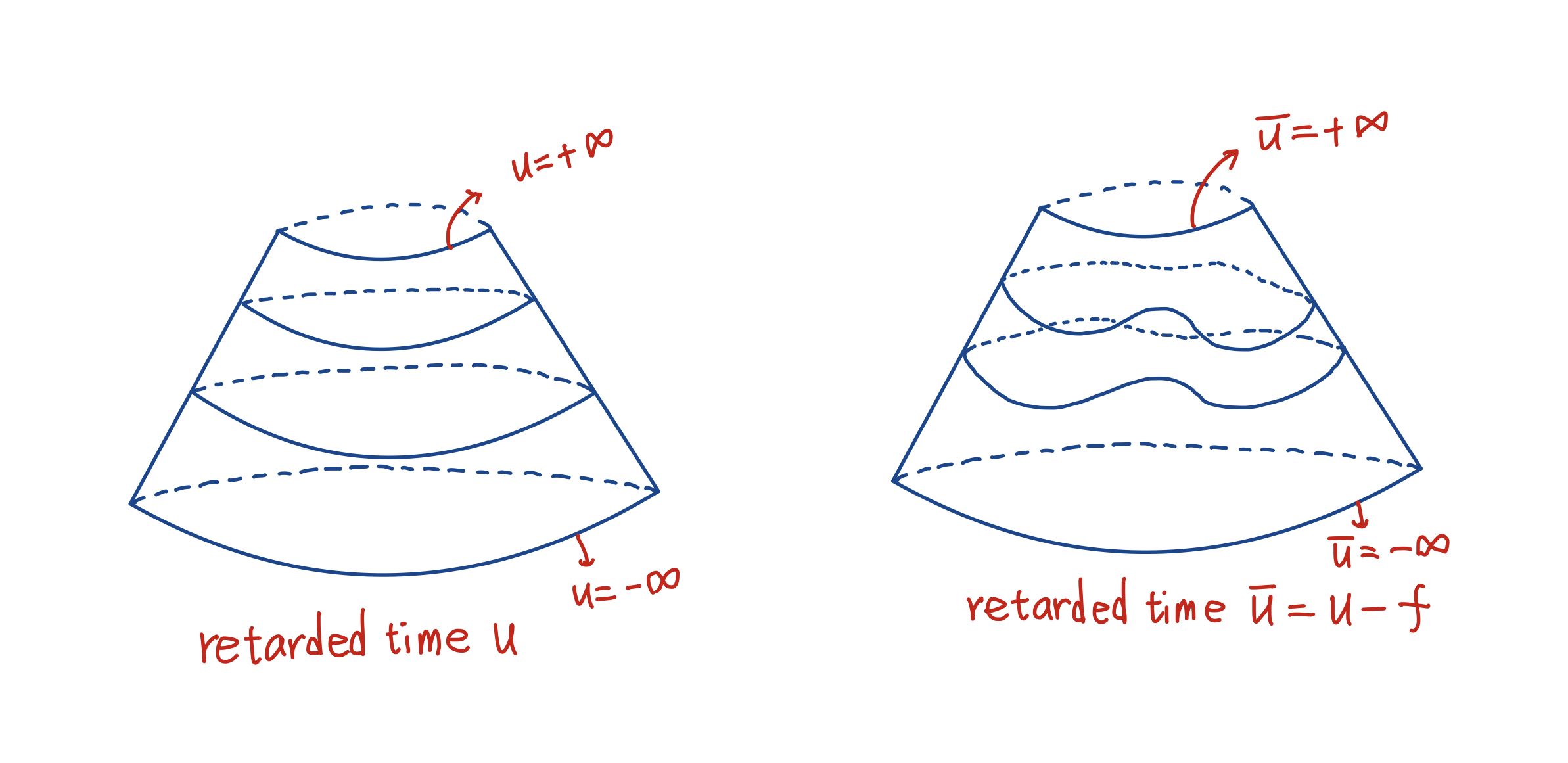}\end{figure}

The standard formulae for the Bondi-Sachs energy-momentum \cite{BVM, Sachs} at a $u$ cut along $\mathscr{I}^+$ are 
\begin{equation}\label{EM} E (u) = \int_{S^2} 2m (u, \cdot), P^k (u) = \int_{S^2} 2m (u, \cdot) \tilde{X}^k, k=1, 2, 3\end{equation}
where $\tilde{X}^k, k=1, 2 ,3$ are the standard coordinate functions on $\mathbb{R}^3$ restricted to $S^2$. 

The vacuum Einstein equation  implies (see \cite{CJK, MW})
\begin{equation}\label{du_m} \partial_u m=-\frac{1}{8}N_{AB}N^{AB}+\frac{1}{4}\na^A\na^B N_{AB}\end{equation} from which the Bondi mass loss formula follows:

\[\partial_u E(u)=-\int_{S^2}  N_{AB}N^{AB} \leq 0.\]
We assume $\mathscr{I}^+$ extends from $i^0$ ($u=-\infty$) to $i^+$ ($u=+\infty$) and that there exists a constant $\varepsilon>0$ such that
\begin{equation}\label{news_decay}
N_{AB}( u,x) = O(|u|^{-1-\varepsilon}) \mbox{ as } u \rw \pm\infty.
\end{equation}

Therefore the total flux of $E(u)$ is 
\begin{equation}\label{flux} \lim_{u\rightarrow +\infty} E(u)-\lim_{u\rightarrow -\infty} E(u)=\int_{-\infty}^{+\infty} \int_{S^2}  N_{AB}N^{AB}\end{equation}

For a supertranslation \eqref{coord_change}, it is known that the news $ \bar{N}_{AB}(\bar u,x)$  in the $(\bar u, \bar{x})$ coordinate system are related to the news $ {N}_{AB}(u,x)$  in the $(u, x)$ coordinate system through:
\begin{equation}\label {news}\bar{N}_{AB}(\bar u,x) = N_{AB}(\bar u+f(x),x).\end{equation} See \cite[(C.117) and (C.119)]{CJK} for example. Together with \eqref{news_decay} and \eqref{flux}, the supertranslation invariance of the total flux of mass \eqref{super-inv} follows. The total energy radiated away is well-defined.

\subsection{Supertranslation invariance of total flux of angular momentum }
The definition of angular momentum turns out to be more subtle, and relies on higher order expansion of $W^A$:
\[W^A=\frac{1}{2r^{2}} {\nabla}^D C_D^A+\frac{1}{r^3}(\frac{2}{3} N^A-\frac{1}{16} {\nabla}^A (C_{DE} C^{DE})-\frac{1}{2} C_B^{\,\,A} {\nabla}^D C_{D}^{\,\,B})+O(r^{-4}).\] 
The one-form $N_A (u, x)=\sigma_{AB} N^B (u, x)$ is called the {\it angular momentum aspect}, which appears in  all existing definitions of angular momentum at $\mathscr{I}^+$.
However, the transformation formula of $N_A$ under supertranslation is rather complicated.
Indeed, the news $N_{AB}$ is the only quantity at $\mathscr{I}^+$ that transforms in a simple manner. See \eqref{supertranslation} for the transformation formulae of $m$ and $C_{AB}$ under a supertranslation. This presents great difficulty in finding a well-defined angular momentum definition at null infinity and none of the existing definitions are known to be supertranslation invariant/covariant. 
There were efforts to eliminate supertranslation ambiguity by choosing special
foliations of $\mathscr{I}^+$, for example,  ``nice sections"by Moreschi \cite{Moreschi},  or ``preferred cuts" by Rizzi \cite{Rizzi}. However, they work only in special cases which seem to exclude black hole formations \cite[Reference 14]{Rizzi}. There were also efforts to identify special conditions on null infinity that would make known definitions of angular momentum supertranslation invariant \cite{ADK1, ADK2}.

In 1982, Penrose \cite{Penrose2} made the following statement:

The very concept of angular momentum gets ``shifted" by supertranslations and ``it is hard to see in these circumstances how one can rigorously discuss such questions as the angular momentum carried away by gravitational radiation".

 In a Bondi-Sachs coordinate system, on a fixed $u$-slice, consider 2-surfaces $r=r_0$ and let $r_0\rightarrow \infty$. The limit of the CWY quasilocal angular momentum at $\mathscr{I}^+$ on this family of 2-surfaces was evaluated by Keller-Y.-K. Wang-Yau in \cite{KWY}. The answer depends on a new terms $c$ given by the following decomposition of $C_{AB}$ and has never occurred in any previous definition of angular momentum.  We consider the decomposition of $C_{AB}$ into 
\begin{equation}\label{CAB}C_{AB}=\nabla_A\nabla_B c-\frac{1}{2} \sigma_{AB} \Delta c+\frac{1}{2}(\epsilon_A^{\,\,\,\, E} \nabla_E \nabla_B \underline{c}+\epsilon_B^{\,\,\,\, E} \nabla_E \nabla_A \underline{c})\end{equation} where $\epsilon_{AB}$ denotes  the volume form of $\sigma_{AB}$.  The functions $c=c(u, x)$ and $\underline{c}=\underline{c}(u, x) $ are called the closed and co-closed potentials of $C_{AB}(u, x)$, repsectively. They are chosen to be of $\ell\geq 2$ harmonic modes and thus such a decomposition is unique.

The term  $c=c(u, x)$ arises in the CWY angular momentum through solving the optimal isometric embedding equation which corresponds to $\Delta(\Delta+2) c=\nabla^A\nabla^B C_{AB}$.

The asymptotic symmetry of $\mathscr{I}^+$ consists of the BMS fields \cite{Sachs}.  We say a BMS field $Y$ is a {\it rotation BMS field} if in a Bondi-Sachs coordinate system
$(u, x)$, 
\begin{equation}\label{rotation_BMS}Y=\hat{Y}^A\frac{\partial}{\partial x^A}\end{equation} where $\hat{Y}^A(x)$ is a rotation Killing field on $S^2$.

The limit of the CWY angular momentum \cite{KWY} is shown to be

\begin{definition}
For a rotation BMS field $Y$ that is tangent to $u$ cuts on $\mathscr{I}^+$, the  angular momentum of a $u$ cut is defined to be :
\begin{align}\label{AM}
 J (u, Y)= \int_{S^2}  \hat{Y}^A \lt( N_{A} -\frac{1}{4}C_{AB}\nabla_{D}C^{DB} - c\nabla_{A}m   \rt) (u, \cdot), 
\end{align}   where $c=c(u, x)$ is given by \eqref{CAB}.
\end{definition}

In comparison,  the classical angular momentum defined by  Dray-Streubel \cite{DS} is 

\[\tilde{J}=\int_{S^2} \hat{Y}^A(N_A-\frac{1}{4} C_{A}^{\,\,\, D} \nabla^B C_{DB}).\]

The last term $ - \int_{S^2} c\nabla_{A}m   (u, \cdot)$ in the definition of $J$ indeed corresponds to the reference term in the Hamiltonian theory.

\subsection{The proof of supertranslation invariance}
Suppose  $(\bar{u}, \bar{x})$ is another Bondi-Sachs coordinate system, we define similarly: 
\begin{equation}
{J} (\bar{u}, \bar{Y})=\int_{S^2}    \hat{Y}^A \left(\bar{N}_A-\frac{1}{4}\bar{C}_{A}^{\,\,\,\,D}\bar{\nabla}^B \bar{C}_{DB}     - \bar{c}\bar{\nabla}_{A}
\bar{m}         \right) (\bar u, \cdot),
\end{equation} where $\bar{Y}=\hat{Y}^A \frac{\partial}{\partial \bar{x}^A}$ is a rotation BMS field that is tangent to $\bar{u}$ cuts.

 Two rotation BMS fields $\bar{Y}$ and ${Y}$ are said to be related by the supertranslation $f$ if, in the $(u, x)$ coordinate system\begin{equation}\label{BMS_related}\bar{Y}=\hat{Y}^A\frac{\partial}{\partial x^A}+\hat{Y}(f)\frac{\partial}{\partial u}
\text{ and }  Y=\hat{Y}^A \frac{\partial}{\partial x^A}\end{equation} for a rotation Killing field $\hat{Y}$ on $S^2$. In this case, $\bar{Y}$ is tangent to the $\bar{u}$ cuts while ${Y}$ is tangent to the ${u}$ cuts.

 The total flux of the angular momentum \eqref{AM} is defined to be \[\delta {J} (Y)=\lim_{u\rightarrow +\infty} {J}(u, Y)-\lim_{u\rightarrow -\infty} {J}(u, Y).\] Under a supertranslation, 
  \[\delta {J} (\bar{Y})=\lim_{\bar{u}\rightarrow +\infty} {J}(\bar{u}, \bar{Y})-\lim_{\bar{u}\rightarrow -\infty} {J}(\bar{u}, \bar{Y}).\]

When two Bondi-Sachs coordinates are related by a supertranslation, one shows that (see \eqref{supertranslation} below) \begin{equation}m(+)=\lim_{u\rightarrow + \infty} m(u, x)=\lim_{\bar{u} \rightarrow + \infty} \bar{m}(\bar{u}, x) \text{ and } m(-)=\lim_{u\rightarrow  -\infty} m(u, x)=\lim_{\bar{u} \rightarrow - \infty} \bar{m}(\bar{u}, x)\end{equation} are two functions on $S^2$ that are invariant under supertranslation.

\begin{theorem} Under condition \eqref{news_decay}, suppose two Bondi-Sachs coordinate systems  are related by a supertranslation $f$, and $Y$ and $\bar{Y}$ are rotation BMS fields related by $f$ \eqref{BMS_related}. Then
\begin{equation}\label{flux_invariant}\delta {J}(\bar{Y})-\delta {J}(Y)= -\int_{S^2} \lt(2 f_{\ell\leq 1} \hat{Y}^A\nabla_A(m(+)-m(-)) \rt), \end{equation}
where $f=f_{\ell\leq 1}+f_{\ell\geq 2}$ is the decomposition into the corresponding harmonic modes.
\end{theorem}

\begin{proof}
The vacuum Einstein equation implies (see \cite{CJK, MW}):
\begin{equation}\label{Einstein} \begin{split}
\partial_u N_A &= \na_A m -\frac{1}{4}\nabla^D(\nabla_D \nabla^E C_{EA}-\nabla_A \nabla^E C_{ED}) \\
&\quad +\frac{1}{4}\nabla_A(C_{BE} N^{BE})-\frac{1}{4}\nabla_B (C^{BD} N_{DA})+\frac{1}{2} C_{AB}\nabla_D N^{DB}.\end{split}\end{equation}
A rotation Killing field $\hat{Y}$ satisfies $\nabla_A \hat{Y}^A=0$ and $\nabla^A\hat{Y}^D+\nabla^D \hat{Y}^A=0$ and thus
\[ \int_{S^2} \hat{Y}^A \partial_uN_A= \int_{S^2} \hat{Y}^A\lt[-\frac{1}{4}\nabla_B (C^{BD} N_{DA})+\frac{1}{2} C_{AB}\nabla_D N^{DB}\rt].\]

Therefore,
\begin{equation}\label{du_classical_AM}\begin{split}&\partial_u  \tilde{J}(u, Y)=\frac{1}{4}\int_{S^2}  \hat{Y}^A  \lt[C_{AB}\nabla_{D}N^{BD} -N_{AB}\nabla_{D}C^{BD}-\nabla_B (C^{BD} N_{DA}) \rt].\end{split}\end{equation}

According to \eqref{AM}, the total flux $\delta J (Y)$ is thus 
\begin{equation}\begin{split}\label{AM_flux}\delta J(Y)=\delta \tilde{J}(Y) - \lt[\int_{S^2}\hat{Y}^A  c\nabla_{A}m\rt]_{u=-\infty}^{u=+\infty},\end{split}  \end{equation}
where \begin{equation}\label{delta_JY} \delta \tilde{J}(Y)=\frac{1}{4}\int_{-\infty}^\infty \int_{S^2}\hat{Y}^A  \lt[C_{AB}\nabla_{D}N^{BD} -N_{AB}\nabla_{D}C^{BD}-\nabla_B (C_{BD} N_{DA}) \rt] (u, \cdot) du.\end{equation}

For a supertranslation \eqref{coord_change}, it is known that the mass aspect $\bar{m} (\bar u, x)
$, the shear $\bar{C}_{AB}(\bar u,x)$  in the $(\bar u, \bar{x})$ coordinate system are related to the mass aspect ${m} (u, x)
$, the shear ${C}_{AB}(u,x)$ in the $(u, x)$ coordinate system through:
\begin{equation}\label{supertranslation}\begin{split}
\bar{m}(\bar u,x) &= m(\bar u+f,x) + \frac{1}{2} (\na^B N_{BD})(\bar u+f,x) \na^D f \\
&\quad + \frac{1}{4} (\partial_{u}N_{BD})(\bar u+f,x) \na^B f\na^D f + \frac{1}{4} N_{BD}(\bar u+f,x) \na^B\na^D f\\
\bar{C}_{AB}(\bar u,x) &= C_{AB}(\bar u+ f(x),x) - 2 \na_A\na_B f + \Delta f \sigma_{AB} \\
\end{split}\end{equation}
See \cite[(C.117) and (C.119)]{CJK} for example. 
%It is understood that $(\na^B N_{BD})(\bar u+f,x)$ means  the evaluation of $(\na^B N_{BD})(u,x)$ at $u=\bar u+f$, or $(\na^B N_{BD})(u,x)|_{u=\bar u+f$.

The decay condition of the news \eqref{news_decay} implies that the limits of the mass aspect and the shear satisfy
\begin{equation}\label{mass_shear_diff} \lim_{\bar u \rw \pm\infty} \bar {m}(\bar u,x) =  \lim_{u \rw \pm\infty}  m(u,x),
\lim_{\bar u\rw\pm\infty} \bar{C}_{AB}(\bar u,x) = \lim_{u\rw \pm\infty} C_{AB}(u,x) -2\na_A\na_B f + \Delta f \sigma_{AB}. 
\end{equation}
In particular, the limit of the potential $c$ satisfies:
\[\lim_{\bar u\rw\pm\infty} c(\bar u,x) = \lim_{u\rw \pm\infty} c(u,x) -2 f_{\ell\geq 2}. \]  It follows that the contribution from the second term of \eqref{AM_flux} in the difference $\delta {J} (\bar{Y})- \delta {J} (Y)$ is \begin{equation}\label{second}\int_{S^2} [2f_{\ell\geq 2} \hat{Y}^A \na_A (m(+)-m(-))].\end{equation}

On the other hand, the contribution from the first term in \eqref{AM_flux}, or $\delta \tilde{J} (\bar{Y})-\delta \tilde{J} (Y)$, after substituting \eqref{supertranslation}, integration by parts, and change of variables, is shown to be \cite{CKWWY}, \[ \frac{1}{4} \int_{-\infty}^{+\infty} \lt[ \int_{S^2} f \hat{Y}^A \na_A \big( N_{BD}N^{BD} - 2\na^B\na^D N_{BD} \big) \rt] du.\] At this point, we invoke the vacuum Einstein equation \eqref{du_m}, change the order of integration (Fubini's Theorem applies due to \eqref{news_decay}), and rewrite the last integral as \[\int_{S^2} [-2f \hat{Y}^A \na_A (m(+)-m(-))].\] In view of \eqref{second}, $\delta J (\bar{Y})-\delta J (Y)$ is given by \eqref{flux_invariant}. 
\end{proof}

By \eqref{EM}, the total flux of linear momentum is $\delta P^k=2\int_{S^2} (m(+)-m(-)) \tilde{X}^k$. It follows that 
\begin{equation}\label{AM_flux_transf}\begin{split}
\delta J (\bar{Y})  &=  \delta J (Y)+  \alpha_i \varepsilon^{ik}_{\;\;\;j} \delta P^j, \text{ if } f = \alpha_0 + \alpha_i \tilde X^i+f_{\ell\geq 2} \text{ and }\hat{Y}^A=\epsilon^{AB} \nabla_B\tilde{X}^k \end{split}
\end{equation} 
In particular, if $f$ is of harmonic mode $\ell\geq 2$, $\delta J (\bar{Y})=\delta J (Y)$ is invariant. Therefore, $\delta J$ is invariant under any $\ell\geq 2$ supertranslation.

For an ordinary translation $f=\alpha_0+\alpha_i \tilde{X}^i$, the above formula can be simplified to 
\[\delta {J} (\bar{Y}) -\delta {J}(Y)= \alpha_i \epsilon^{ik}_{\,\,\,\, j}\delta P^j \text{ for } \hat{Y}^A=\epsilon^{AB}\nabla_B \tilde{X}^k,\] where $\delta P^j$ is the total flux of Bondi-Sachs linear momentum.

This is the classical transformation law of angular momentum when the system is ``ordinarily" translated. The angular momentum carried away by gravitational radiation can be defined free of any ambiguity.

\section{Appendix}

In the appendix, we explicitly write down a Poincar\'e transformation in the Minkowski spacetime as a BMS transformation between Bondi-Sachs coordinates.
 Suppose a Poincar\'e transformation is given between two Cartesian coordinates $\bar{x}^\alpha, \alpha=0, 1, 2, 3$ and $x^\alpha, \alpha=0, 1, 2, 3$ of the Minkowski spacetime with
 \[\bar{x}^\alpha=A^\alpha_\beta x^\beta+T^\alpha,\] such that $A\in O(3, 1)$ and $T^\alpha\in \mathbb{R}^{3,1}$. We shall use the indexes $i, j, k=1, 2, 3$ and $\bar{x}^0=\bar{t}, x^0=t$. To convert to the corresponding Bondi-Sachs coordinates we have $u=t-r$ and $\bar{u}=\bar{t}-\bar{r}$, where $\bar{r}=\sqrt{\sum_{i=1}^3 (\bar{x}^i)^2}$ and $r=\sqrt{\sum_{i=1}^3 (x^i)^2}$. We also denote $\tilde{x}^i=\frac{x^i}{r}$. 
 
 We compute the expansions of $\bar{r}$ and $\bar{u}$ in inverse powers of $r$ in the following. The expansion of $\bar{r}$ is given by:
 
 \[\begin{split}
 \bar{r}&=\sqrt{\sum_{i=1}^3(A^i_\beta x^\beta+T^i)^2}\\
 &=\sqrt{\sum_{i=1}^3[A_0^i(u+r)+r A^i_j \tilde{x}^j+T^i]^2}\\
 &=r \sqrt{\sum_{i=1}^3[A_0^i + A^i_j \tilde{x}^j    +r^{-1}(A_0^i u+T^i)] ^2}\\
 &=r \sqrt{\sum_{i=1}^3(A_0^i + A^i_j \tilde{x}^j)^2    +2 r^{-1}(A_0^i+A_j^i\tilde{x}^j)(A_0^i u+T^i)+O(r^{-2}) }\\
  \end{split}\]
  
  Since $(1, \tilde{x}^1, \tilde{x}^2, \tilde{x}^3)$ is a null vector, so is its image under an $O(3,1)$ action, and thus we have
  \[\sum_{i=1}^3(A_0^i + A^i_j \tilde{x}^j)^2=(A_0^0 + A^0_j \tilde{x}^j)^2.\]
  
  Therefore, we obtain that 
  \[\bar{r}= r (A_0^0 + A^0_j \tilde{x}^j)+\frac{(A_0^i+A_j^i\tilde{x}^j)(A_0^i u+T^i)}{ (A_0^0 + A^0_j \tilde{x}^j)}+O(r^{-1}).\]
 We proceed to compute $\bar{u}$.
   \[\begin{split}\bar{u}&=\bar{x}^0-\bar{r}\\
  &=A_0^0(u+r)+rA_i^0\tilde{x}^i+T^0-\bar{r}.\\
  \end{split}\]
  
 Comparing with the expansion of $\bar{r}$, we derive that the leading coefficient of $r^1$ is zero. The coefficient of the constant order term ($r^0$) is then
 \[ \begin{split}&A_0^0 u+T^0-\frac{(A_0^i+A_j^i\tilde{x}^j)(A_0^i u+T^i)}{ A_0^0 + A^0_j \tilde{x}^j}\\
 &=u\frac{A_0^0(A_0^0 + A^0_j \tilde{x}^j)-A_0^i(A_0^i+A_j^i\tilde{x}^j)}{A_0^0 + A^0_j \tilde{x}^j}+\frac{  T^0(A_0^0 + A^0_j \tilde{x}^j)-T^i(A_0^i+A_j^i\tilde{x}^j)}{A_0^0 + A^0_j \tilde{x}^j}\\
 &=\frac{u}{A_0^0 + A^0_j \tilde{x}^j}+\frac{  T^0(A_0^0 + A^0_j \tilde{x}^j)-T^i(A_0^i+A_j^i\tilde{x}^j)}{A_0^0 + A^0_j \tilde{x}^j}\\
 &=K(u+\tilde f) \end{split}, \] where $\tilde f= T^0(A_0^0 + A^0_j \tilde{x}^j)-T^i(A_0^i+A_j^i\tilde{x}^j)$ is of mode $\ell=1$ and $K=\frac{1}{A_0^0 + A^0_j \tilde{x}^j}$. Note that we use the $O(3,1)$ invariance in the second equality above. An ordinary translation thus corresponds to the case $K=1$ or $A_0^0=1, A^0_j=0, j=1, 2, 3$.

%%  The body

%%  The bibliography

%\begin{thebibliography}{9}
%%  Use \bibitem{r1} or \bibitem[Surname(2010)]{r1} (for authoryear case)
%
%\bibitem{}
%
%\end{thebibliography}

\end{document}